\documentclass[3p,sort&compress]{elsarticle}
\usepackage{amsmath}
\usepackage{color}
\usepackage{amsthm}
\usepackage{amssymb}
\usepackage{titlesec}
\usepackage[thinlines,thiklines]{easybmat}

\usepackage{graphicx}
\usepackage{epsfig}

\theoremstyle{definition}
\theoremstyle{plain}\newtheorem{theorem}{Theorem}
\theoremstyle{definition}
\theoremstyle{definition}
\theoremstyle{plain}
\theoremstyle{plain}\newtheorem{corollary}[theorem]{Corollary}
\theoremstyle{plain}\newtheorem{lemma}[theorem]{Lemma}

\begin{document}
\begin{frontmatter}
\title{Exact quantum algorithms have advantage for almost all Boolean functions}
\author{Andris Ambainis $^{2,3}$}
\author{Jozef Gruska$^{1}$}
\author{Shenggen Zheng$^{1,}$\corref{one}}

 \cortext[one]{Corresponding author.\\ \indent{\it E-mail addresses:} andris.ambainis@lu.lv, gruska@fi.muni.cz, zhengshenggen@gmail.com}

\address{
  $^{1}$Faculty of Informatics, Masaryk University, Brno 60200, Czech Republic\\

  $^2$  Faculty of Computing, University of Latvia,R\=\i ga, LV-1586, Latvia\\

  $^3$ School of Mathematics, Institute for Advanced Study, Princeton, NJ 08540, USA\\
}

\begin{abstract}
It has been proved that  almost  all $n$-bit Boolean functions have  {\em exact classical  query complexity} $n$.
However, the situation  seemed to be very different when we deal with   {\em  exact quantum query complexity}.
In this paper, we prove that  almost  all $n$-bit Boolean functions can be computed by an exact quantum algorithm with less than $n$ queries.
More exactly, we prove that $\mbox{AND}_n$ is the only  $n$-bit Boolean function, up to isomorphism, that requires $n$ queries.
\end{abstract}


\end{frontmatter}

\section{Introduction}

{\em  Quantum query complexity} is the quantum generalization of classical {\em  decision tree complexity}. In this complexity model, an algorithm is charged for ``queries" to the input bits, while any intermediate computation is considered as free (see \cite{BdW02}). For many functions one can obtain large quantum speed-ups
in this model in the case algorithms are allowed a constant small probability of error (bounded error). As the most famous example, Grover's algorithm \cite{Gro96} computes the $n$-bit $\mbox{OR}$ function
with $O(\sqrt {n})$ queries in the bounded error mode, while any classical (also exact quantum) algorithm needs
$\Omega(n)$ queries. More such cases of polynomial
speed-ups are known,  see\cite{Amb07,Bel12,DHHM06}.
For {\em partial functions},
even an exponential speed-up is possible, in case quantum resources are used, see \cite{Shor97,Sim97}.
In  the bounded-error setting, quantum complexity is now relatively well understood. The model of  {\em exact quantum query complexity}, where the algorithms must output the correct answer with
certainty for every input, seems to be more intriguing. It is much more difficult to come up with  exact quantum algorithms that outperform, concerning number of queries,  classical exact algorithms.

Though for partial functions exact quantum algorithms with
exponential speed-up are known (for instance in \cite{AmYa11,BH97,DJ92,GQZ14,ZQ14,GQZ14b,Zhg13}), the results for total functions have been much less spectacular: the best known quantum
speed-up was just by a factor of 2 for many years \cite{CEMM98,FGGS98}. Recently, in a breakthrough result, Ambainis \cite{Amb13} has
presented the first example of a Boolean function $f:\{0,1\}^n\to \{0,1\}$ for which exact quantum algorithms have superlinear advantage over exact classical algorithms.

In  exact classical  query complexity ({\em decision tree complexity}, {\em deterministic query complexity}) model, almost all $n$-bit Boolean functions require $n$  queries \cite{BdW02}. However,  the situation  seemed very different for the case of  exact quantum complexity.
  Montanaro et al.~\cite{MJM11}  proved that  $\mbox{AND}_3$ is the only $3$-bit Boolean function, up to isomorphism, that requires 3 queries and using the semidefinite programming approach,
they numerically\footnote{In their numerical experiments, computation providing correct
result with a probability greater than 0.999 is treated as exact.} demonstrated that  all $4$-bit Boolean functions, with the exception of functions
isomorphic to the $\mbox{AND}_4$ function, have  exact quantum query algorithms using at most 3
queries. They also listed their numerical results for all symmetric Boolean functions
on 5 and 6  bits, up to isomorphism.

In 1998, Beals at al. \cite{BBC+98}  proved,  for any $n$, that $\mbox{AND}_n$ has
exact quantum complexity $n$. Since that time it was an interesting
problem whether $\mbox{AND}_n$ is the only $n$-bit Boolean function, up to
isomorphism, that has  exact quantum complexity  $n$. In
this paper we  approve that this is indeed the case. As a corollary we
get that almost all $n$-bit Boolean functions have  exact quantum
complexity less than $n$.

We prove our main results in four stages. In the first one we give the
proof  for symmetric Boolean functions, in the second one for monotone
Boolean functions and in the third one for the case of read-once
Boolean functions. On this basis we prove in the fourth stage the
general case.
In all four cases proofs used quite different approaches. They
are expected to be of a broader interest since all these special
classes of Boolean functions are of broad interest.

The paper is organized as follows. In Section~2 we introduce some notation concerning Boolean function and query complexity.
In Section~3 we investigate symmetric Boolean functions. In section~4 we investigate monotone Boolean functions. In section~5 we investigate read-once Boolean functions. In Section~6 we prove our main result.
Finally, Section~7 contains a conclusion.

\section{Preliminaries}
We introduce some basic needed  notation in this section. See also \cite{Gru99,NC00} for details on quantum computing and see \cite{BdW02,BBC+98,NS94} for more  on query complexity models and {\em multilinear polynomials}.

\subsection{Boolean functions}
An $n$-bit Boolean function is a function $f:\{0,1\}^n\to \{0,1\}$. We say $f$ is total if $f$ is defined on all inputs. For an input $x\in\{0,1\}^n$, we use $x_i$ to denote its $i$-th bit, so $x=x_1x_2\cdots x_n$. Denote $[n]=\{1,2,\ldots,n\}$.  For $i\in[n]$, we write
\begin{equation}
f_{x_i=b}(x)=f(x_1,\ldots,x_{i-1},b,x_{i+1},\ldots,x_n),
\end{equation}
which is an $(n-1)$ bit Boolean function.
For any $i\in[n]$, we have
\begin{equation}\label{Eq-df(x)}
    f(x)=(1-x_i)f_{x_i=0}(x)+x_if_{x_i=1}(x).
\end{equation}

We say that two Boolean functions $f$ and $g$ are {\em query-isomorphic} (by convenience,  isomorphic will mean query-isomorphic in this paper) if they are equal up to negations
and permutations of the input variables, and negation of the output variable.  This relationship is sometimes known as NPN-equivalence \cite{MJM11}.

We will use the sign $(\neg)$ for a possible negation. For example, $\mbox{AND}((\neg)x_1,x_2)$ can denote $x_1\wedge x_2$ or $\neg x_1\wedge x_2$.
We use $|x|$ to denote the Hamming weight of $x$ (its number of 1's).

\vspace*{12pt}
\noindent
{\bf Definition~1:}
We call a Boolean function $f:\{0,1\}^n\to \{0,1\}$ symmetric if $f(x)$  depends only on $|x|$.

\vspace*{12pt}
\noindent

  An $n$-bit symmetric Boolean function $f$ can be fully described by a vector $(b_0,b_1,\ldots,b_n)\linebreak[0]\in\{0,1\}^{n+1}$, where $f(x)=b_{|x|}$, i.e. $b_k$ is the value of $f(x)$ for $|x|=k$ \cite{ZGR97}.

For $x,y\in\{0,1\}^n$, we will write $x\preceq y$ if $x_i\leq y_i$ for all $i\in[n]$. We will write $x\prec y$ if $x\preceq y$ and $x\neq y$.

\vspace*{12pt}
\noindent
{\bf Definition~2:}
We call a Boolean function  $f:\{0,1\}^n\to \{0,1\}$   monotone   if  $f(x)\leq f(y)$ holds whenever $x\preceq y$.
\vspace*{12pt}
\noindent

Monotonic Boolean functions are precisely those that can be defined by an expression combining the input bits (each of them may appear more than once) using only the operators $\wedge$ and $\vee$ (in particular $\neg$ is forbidden).
Monotone Boolean functions have many nice properties. For example  they have a unique prime
conjunctive normal form (CNF)  and a unique prime disjunctive normal form (DNF) in which no negation occurs \cite{EMG08}.

 Let $f:\{0,1\}^n\to \{0,1\}$ be a monotone Boolean function, $f$ has a prime CNF
\begin{equation}
    f(x)=\bigwedge_{I\in C}\bigvee_{i\in I} x_i,
\end{equation}
 where  $C$ is the set of some $I\subseteq[n]$.
Similarly,  $f$ has a prime DNF
\begin{equation}
    f(x)=\bigvee_{J\in D}\bigwedge_{j\in J} x_j,
\end{equation}
where $D$ is
the set of some $J\subseteq[n]$.

\vspace*{12pt}
\noindent
{\bf Definition~3:}
A read-once Boolean function is a Boolean function that can be represented by  a Boolean formula in which each variable appears exactly once.
\vspace*{12pt}
\noindent

For example $f(x_1,x_2,x_3)=(x_1\vee x_2)\wedge (\neg x_3)$ is a $3$-bit read-once Boolean function and $f'(x_1,x_2,x_3)=(x_1\vee x_2)\wedge (\neg x_1\vee \neg x_3)$ is not read-once.

A Boolean formula over
the standard basis $\{\wedge,\vee,\neg \}$ can be represented by a binary tree where each internal node is labeled
with $\wedge$ or $\vee$, and each leaf is labeled with a literal, that is, a Boolean variable
or its negation. The size of a formula is the number of leaves.

\vspace*{12pt}
\noindent
{\bf Definition~4:}
 The formula
size of a Boolean function $f$, denoted $L(f)$, is the size of the smallest formula which computes $f$.
\vspace*{12pt}
\noindent

A read-once Boolean function is a function $f$ such that $L(f)=n$ and $f$ depends on all of its $n$ variables.

\subsection{Exact query complexity models}

An exact classical (deterministic) query algorithm for computing a Boolean function $f:\{0,1\}^n\to \{0,1\}$ can be described
by a decision tree. A decision tree $T$ is a rooted binary tree where each internal vertex
has exactly two children, each internal vertex is labeled with a variable $x_i$ and each leaf is labeled with a value 0 or 1. $T$ computes a Boolean function $f$ as follows: Start at the root. If this is a leaf then stop and the output of the tree is the value of the leaf.  Otherwise, query the variable $x_i$ that labels the root. If $x_i=0$, then recursively evaluate the left subtree, if $x_i=1$ then recursively evaluate the right subtree. The output of the tree is the value of the leaf that is reached at the end of this
process.
The depth
of $T$ is the maximal length of a path from the root to a leaf (i.e. the worst-case number of queries
used on any input). The {\em exact classical query complexity} (deterministic
query complexity, decision tree complexity) $D(f)$ is the minimal depth over all decision trees computing $f$.

Let $f:\{0,1\}^n\to \{0,1\}$ be a Boolean function and $x = x_1x_2\cdots x_n$ be an input
bit string.   Each exact quantum query algorithm for $f$
works in a Hilbert space with some fixed basis, called standard. It starts in a
fixed starting state, then performs on it a sequence of  transformations
$U_1$, $Q$, $U_2$, $Q$, \ldots, $U_t$, $Q$, $U_{t+1}$.
Unitary transformations $U_i$ do not depend on
the input bits, while $Q$, called the {\em query transformation}, does,
in the following way. Each of the basis states corresponds to either one or none
of the input bits. If the basis state $|\psi\rangle$ corresponds to the $i$-th
input bit, then $Q|\psi\rangle=(-1)^{x_i}|\psi\rangle$. If it does not correspond to any
input bit, then $Q$ leaves it unchanged: $Q|\psi\rangle=|\psi\rangle$. Finally, the algorithm performs a  measurement in the standard basis.
Depending on the result of the measurement, the algorithm outputs either 0 or 1
which must be equal to $f(x)$. The {\em  exact quantum query complexity}
$Q_E(f)$ is the minimum number of queries used by any quantum algorithm which
computes $f(x)$ exactly for all $x$.

Note that if Boolean functions $f$ and $g$ are isomorphic, then $D(f)=D(g)$ and $Q_E(f)=Q_E(g)$.

According to Eq.~(\ref{Eq-df(x)}), if we  query $x_i$ first, suppose that $x_i=b$, then we can compute $f_{x_i=b}(x)$ further. Therefore, for any $i\in[n]$, we have
\begin{equation}\label{Eq-n-1ton}
Q_E(f)\leq \max\{Q_E(f_{x_i=0}),Q_E(f_{x_i=1})\}+1.
\end{equation}

\subsection{Some special functions and their  exact quantum query complexity}

Symmetric, monotone and read-once Boolean functions were  well studied in query complexity \cite{BdW02}. The well known Grover's algorithm \cite{Gro96} computes  $\mbox{OR}_n$, which is symmetric, monotone and read-once.  Read-once functions are also well investigated \cite{BS04,SW86,San95}.

Some symmetric functions and their exact quantum query complexity that we will refer to in this paper are as follows:
\begin{enumerate}
\item $\mbox{OR}_n(x)=1$  iff $|x|\geq 1$.  $Q_E(\mbox{OR}_n)=n$ \cite{BBC+98}.
  \item $\mbox{AND}_n(x)=1$  iff $|x|=n$. $Q_E(\mbox{AND}_n)=n$ \cite{BBC+98}.
  \item $\mbox{PARITY}_n(x)=1$ iff $|x|$ is odd. $Q_E(\mbox{PARITY}_n)=\lceil\frac{n}{2}\rceil$ \cite{CEMM98,FGGS98}.
  \item $\mbox{EXACT}_n^{k}(x)=1$  iff $|x|=k$.  $Q_E(\mbox{EXACT}_n^{k})=\max\{k,n-k\}$ \cite{AISJ13}.
  \item  $\mbox{Th}_{n}^{k}(x)=1$ iff $|x|\geq k$. $Q_E(\mbox{Th}_n^{k})=\max\{k,n-k+1\}$  \cite{AISJ13}.
\end{enumerate}

$\mbox{OR}_n$ is isomorphic to $\mbox{AND}_n$ since
 \begin{equation}
 \neg\mbox{OR}_n(\neg x_1,\neg x_2,\ldots,\neg x_n)=\mbox{AND}_n(x_1,x_2,\ldots,x_n).
  \end{equation}

Some other functions and their exact quantum query complexity that we will refer to in this paper are as follows:
\begin{enumerate}
  \item  $\mbox{NAE}_{n}(x)=1$ iff there exist $i,j$ such that $x_i\neq x_j$.  $Q_E(\mbox{NAE}_{n})\leq n-1$.
  \item $f(x_1,x_2,x_3)=x_1\wedge(x_2\vee x_3)$. Its exact quantum  query complexity is 2 \cite{MJM11}.
  \end{enumerate}
   It is easy to prove that $Q_E(\mbox{NAE}_{n})\leq n-1$ since
   \begin{equation}
   \mbox{NAE}_{n}(x_1,\ldots,x_n)=(x_1\oplus x_2)\vee(x_2\oplus x_3)\cdots\vee(x_{n-1}\oplus x_n).
\end{equation}

\subsection{Multilinear polynomials}

 Every Boolean function $f:\{0,1\}^n\to \{0,1\}$ has  a unique representation as an $n$-variate multilinear polynomial  over the reals, i.e., there exist real coefficients $a_S$ such that
 \begin{equation}
 f(x_1,\ldots,x_n)=\sum_{S\subseteq [n]} a_S \prod_{i\in S} x_i.
 \end{equation}
  The degree of $f$  is the degree of its largest monomial:
 $deg(f)=\max\{|S|:a_S\neq 0\}$.

  For example, $\mbox{AND}_2(x_1,x_2)=x_1\cdot x_2$ and $\mbox{OR}_2(x_1,x_2)=x_1+x_2-x_1\cdot x_2$.

 $\textrm{deg}(f)$ gives a lower bound on $D(f)$. Indeed, it holds

\begin{lemma} {\em \cite{BdW02}}\label{D(f)-geq-Deg(f)}
$D(f)\geq \textrm{deg}(f)$.
\end{lemma}

\section{Symmetric Boolean functions}

 \begin{theorem}
Let $f:\{0,1\}^n\to \{0,1\}$ be a symmetric Boolean function. $Q_E(f)=n$ iff $f$ is isomorphic to $\mbox{AND}_n$.
\end{theorem}

\begin{proof}
If $f$ is isomorphic to $\mbox{AND}_n$, then $Q_E(f)=n$ \cite{BBC+98}.

An $n$-bit symmetric Boolean function  can be fully described by a vector $(b_0,b_1,\ldots,b_n)\in\{0,1\}^{n+1}$, where $f(x)=b_{|x|}$, i.e. $b_k$ is the value of $f(x)$ for $|x|=k$.

\begin{table}[http]
\caption{Exact quantum query complexity for $3$-bit symmetric functions.}\label{T1}      % Give a unique label
\begin{center}
\begin{tabular}{|l|l|l|}

  \hline
  $(b_0,b_1,b_2,b_3)$  &Type of function& Query complexity\\
  \hline
  0 0 0 0  & Constant function & 0\\
  \hline
  0 0 0 1  & $\mbox{AND}_3$ & 3\\
  \hline
  0 0 1 0  &$\mbox{EXACT}_3^{2}$ & 2\\
  0 0 1 1  &$\mbox{Th}_3^{2}$ & 2\\
  0 1 0 0 &$\mbox{EXACT}_3^{1}$ & 2 \\
  0 1 0 1  & $\mbox{PARITY}_3$ & 2\\
  0 1 1 0 & $\mbox{NAE}_3$ & 2 \\
  \hline
  0 1 1 1  & Isomorphic to $\mbox{AND}_3$& 3\\
  \hline
  1 0 0 0  & Isomorphic to $\mbox{AND}_3$& 3\\
  \hline
  1 0 0 1  & Isomorphic to $\mbox{NAE}_3$& 2\\
  1 0 1 0  &  Isomorphic to $\mbox{PARITY}_3$& 2\\
  1 0 1 1  & Isomorphic to $\mbox{EXACT}_3^{1}$& 2\\
  1 1 0 0  & Isomorphic to $\mbox{Th}_3^{2}$ & 2\\
  1 1 0 1  & Isomorphic to $\mbox{EXACT}_3^{2}$& 2\\
  \hline
  1 1 1 0 & Isomorphic to $\mbox{AND}_3$& 3\\
  \hline
  1 1 1 1  & Constant function & 0\\
  \hline
\end{tabular}
\end{center}
\end{table}

 Table~\ref{T1} contains all 3-bit Boolean
functions and their exact quantum query complexity.
Four 3-bit  Boolean functions that achieve 3 queries are those that can be  described by one of the following vectors:  $(0,0,0,1),(0,1,1,1),(1,0,0,0), (1,1,1,0)$. They are isomorphic to  $\mbox{AND}_3$.

We claim  that  only $n$-bit  Boolean functions that can be described by one of the following vectors $(0,\ldots,0,1),\linebreak[0](0,1,\ldots,1),\linebreak[0](1,0,\ldots,0),\linebreak[0](1,\ldots,1,0)$, which are isomorphisms of $\mbox{AND}_n$,   that can achieve $n$ queries.
We prove this claim by an induction on $n$ as follows:

{\bf BASIS}: The result holds clearly for $n=3$.

{\bf INDUCTION}: Suppose the result holds for $n=k$ ($\geq 3$). We will prove that the result holds also for $n=k+1$.
We use vector $(b_0,b_1,\ldots,b_k,b_{k+1})$ to describe the function $f(x_1,\cdots,x_k,x_{k+1})$.
Since
\begin{equation}
Q_E(f)\leq \max\{Q_E(f_{x_1=0}),Q_E(f_{x_1=1})\}+1,
\end{equation}
we just need to consider the case that at least one of the functions $f_{x_1=0}$ and $f_{x_1=1}$ is isomorphic to $\mbox{AND}_k$. For other cases we have $Q_E(f)<k+1$.

\begin{table}[http]
\caption{Exact quantum query complexity for $(k+1)$-bit symmetric Boolean functions. }\label{T2}
\begin{center}
\begin{tabular}{|l|l|l|}
 \hline
  $b_0b_1\ldots,b_k,b_{k+1}$ & Type of function & Query complexity \\
      \hline
  $(0,|0,\ldots,0,1)$ & $\mbox{AND}_{k+1}$ & $k+1$ \\
    \hline
  $(0,|0,1,\ldots,1)$ & $\mbox{Th}_{k+1}^{2}$ & $k$ \\
    \hline
  $(0,|1,0,\ldots,0)$ & $\mbox{EXACT}_{k+1}^1$ & $k$ \\
    \hline
  $(0,|1,\ldots,1,0)$ & $\mbox{NAE}_{k+1}$ & $<k+1$\\
    \hline
      \hline
  $(1,|0,\ldots,0,1)$ &Isomorphic to $\mbox{NAE}_{k+1}$ & $<k+1$ \\
    \hline
  $(1,|0,1,\ldots,1)$ & Isomorphic to $\mbox{EXACT}_{k+1}^1$ & $k$ \\
    \hline
  $(1,|1,0,\ldots,0)$& Isomorphic to $\mbox{Th}_{k+1}^{2}$ &$k$\\
    \hline
  $(1,|1,\ldots,1,0)$ & Isomorphic to $\mbox{AND}_{k+1}$ & $k+1$ \\
    \hline
      \hline
  $(0,\ldots,0,1,|0) $& $\mbox{EXACT}_{k+1}^{k}$ & $k$ \\
    \hline
  $(0,1,\ldots,1,|0) $&   $\mbox{NAE}_{k+1}$& $<k+1$ \\
    \hline
  $(1,0,\ldots,0,|0)$ & Isomorphic to $\mbox{AND}_{k+1}$ & $k+1$ \\
    \hline
  $(1,\ldots,1,0,|0) $& Isomorphic to $\mbox{Th}_{k+1}^{k}$ & $k$ \\
    \hline
      \hline
  $(0,\ldots,0,1,|1)$ &  $\mbox{Th}_{k+1}^{k}$ & $k$ \\
    \hline
  $(0,1,\ldots,1,|1) $&  Isomorphic to $\mbox{AND}_{k+1}$ &$k+1$ \\
    \hline
  $(1,0,\ldots,0,|1) $& Isomorphic to $\mbox{NAE}_{k+1}$& $<k+1$ \\
    \hline
  $(1,\ldots,1,0,|1)$ &Isomorphic to $\mbox{EXACT}_{k+1}^{k}$ & $k$ \\
  \hline
\end{tabular}

\end{center}
\end{table}

There are three cases we have to consider according to the value of $b$.

{\bf Case 1}  $b=(0,\ldots,0,1)$. In this case $f=\mbox{AND}_{k+1}$.

{\bf Case 2}  $b=(1,0,\ldots,0)$. In this case  $f$ is isomorphic to  $\mbox{AND}_{k+1}$.

{\bf Case 3} Otherwise,   $f_{x_1=0}$  can be described by the vector $(b_0,b_1,\ldots,b_{k})$ and  $f_{x_1=1}$  can be described by the vector $(b_1,\ldots,b_{k},b_{k+1})$.
Thus we just need to consider Boolean functions that can be described by vector $b=(b_0,b_1,\ldots,b_k,b_{k+1})$ such that
one of the following vectors
\begin{equation}
(\overbrace{0,\ldots,0}^{k},1),\linebreak[0](0,\overbrace{1,\ldots,1}^{k}),\linebreak[0](1,\overbrace{0,\ldots,0}^k),(\overbrace{1,\ldots,1}^k,0)
\end{equation}
is its prefix or suffix\footnote{Let  $b=(b_0,b_1,\ldots, b_{k+1})$. We say that  $(b_0,\ldots, b_{k})$  is a prefix of $b$ and  $(b_1,\ldots, b_{k+1})$  a suffix of $b$.}. There are 16 such Boolean functions and their query complexity are listed in Table~\ref{T2}.

According to   Table~\ref{T2}, only  $(k+1)$-bit Boolean functions which are isomorphic to $\mbox{AND}_{k+1}$ require $k+1$ queries. Thus, the theorem has been proved.
\end{proof}

It is mentioned in \cite{MJM11,Aar03} that  all non-constant $n$-bit  symmetric Boolean functions have exact classical  complexity  $n$. We give now a rigorous proof of that.

\begin{theorem}
If $f:\{0,1\}^n\to \{0,1\}$ is a    non-constant symmetric function, then $D(f)=n$.
\end{theorem}
\begin{proof}
Suppose $f$ can be described by the vector $(b_0,b_1,\ldots,b_n)\in\{0,1\}^{n+1}$. Since $f$ is non-constant, there exists a $k\in[n]$ such that $b_{k-1}\neq b_{k}$. If the first $k-1$ queries return $x_i=1$ and the next $n-k$ queries return $x_i=0$, then we will need to query the last variable as well.
\end{proof}

\section{Monotone Boolean functions}

\begin{theorem}
Let $f:\{0,1\}^n\to \{0,1\}$ be a monotone Boolean  function. $Q_E(f)=n$ iff $f$ is isomorphic to $\mbox{AND}_n$.
\end{theorem}
\begin{proof}
Obviously,  $\mbox{AND}_n(x)$ and $\mbox{OR}_n(x)$ are the only two $n$-bit monotone Boolean functions that are isomorphic to $\mbox{AND}_n(x)$.
If $f$ is isomorphic to $\mbox{AND}_n(x)$, then $Q_E(f)=n$  \cite{BBC+98}. We prove the other direction by an induction on $n$.

{\bf BASIS}: Case $n=2$,  $\mbox{AND}_2(x_1,x_2)$ is the only $2$-bit function, up to isomorphism, that  requires 2 queries. Therefore the result holds for $n=2$.

{\bf INDUCTION}: Suppose the result holds for all $n\leq k$, we  prove that the result holds also for $n=k+1$ in the following way.

For any $i\in[k+1]$, if   $Q_E(f_{x_i=0})<k$ and  $Q_E(f_{x_i=1})<k$, then $Q_E(f)\leq \max\{Q_E(f_{x_i=0}),\linebreak[0]Q_E(f_{x_i=1})\}+1<k+1$.
Therefore, we  need  to consider only the case that
at least one of functions  $f_{x_i=0}$ and  $f_{x_i=1}$ requires $k$ quires. There are two  such cases:

{\bf Case 1:}
$Q_E(f_{x_1=1})=k$. According to the assumption, $f_{x_1=1}$ is isomorphic to $\mbox{AND}_k$. There are now two subcases  to
    consider:

{\bf Case 1a:} $f_{x_1=1}(x)=\mbox{OR}_k(x_2,\cdots,x_{k+1})=\mbox{OR}_k(x_{-1})$ (For convenience, we write $x_{-i}=x_1,\ldots, x_{i-1},x_{i+1},\linebreak[0]\ldots x_{k+1}$). Let us consider the CNF of $f$:
\begin{equation}
f(x)=\bigwedge_{I\in C}\bigvee_{i\in I} x_i=\left(\bigwedge_{I\in C,1\in I}\bigvee_{i\in I} x_i\right)\wedge\left(\bigwedge_{I\in C,1\not\in I}\bigvee_{i\in I} x_i\right)
.
\end{equation}
Therefore,
\begin{equation}
f(x)=(x_1\vee g_1(x_{-1}))\wedge \mbox{OR}_k(x_{-1}),
\end{equation}
where $x_1\vee g_1(x_{-1})=\left(\bigwedge_{I\in C,1\in I}\bigvee_{i\in I} x_i\right)$ and $g_1$ is also a monotone  function.
So we have $f(x)=1$ for any $x$ such that $10\cdots 0\prec x$ and $f(x)=0$ for any $x$ such that $x\preceq 10\cdots 0 $.

Let us consider now two subcases. Namely $f_{x_2=1}$ and $f_{x_2=0}$.
Since $10\cdots0\preceq 10\cdots 0$, we have $f(10\cdots0)\linebreak[0]=0$ and $f_{x_2=0}(x)\neq \mbox{OR}_k(x_{-2})$.
Since $10\cdots 0\prec 1010\cdots 0$, we have $f(1010\cdots 0)\linebreak[0]=1$ and $f_{x_2=0}(x)\neq \mbox{AND}_k(x_{-2})$.
Now we have $Q_E(f_{x_2=0})<k$ and therefore $Q_E(f_{x_2=1})\linebreak[0]=k$. Since $10\cdots 0\prec 110\cdots 0$, we have $f(110\cdots 0)=1$ and $f_{x_2=1}(x)\neq \mbox{AND}_k(x_{-2})$.
Therefore, $f_{x_2=1}(x)=\mbox{OR}_k(x_{-2})$.
 Using a similar argument, we can prove that for any $i\geq 2$, $f_{x_i=1}(x)=OR_k(x_{-i})$.
 Hence, for any $i\in[k+1]$, we have
\begin{equation}
f(x)=(x_i\vee g_i(x_{-i}))\wedge \mbox{OR}_k(x_{-i}).
\end{equation}
So  $f(x)=1$ for any $x$ such that $y\prec x$
 and $f(x)=0$ for any $x$ such that $x\preceq y$, where $y_i=1$ and $y_j=0$ for any $j\neq i$. It is not hard to see that in this case $f(x)=\mbox{Th}_{k+1}^2(x)$ and therefore $Q_E(f)=k$.

{\bf Case 1b:}    $f_{x_1=1}(x)=\mbox{AND}_k(x_{-1})$.  Let us consider the CNF of $f$.
  We have,
\begin{equation}
f(x)=(x_1\vee g'(x_{-1}))\wedge \mbox{AND}_k(x_{-1}),
\end{equation}
where $g'(x_{-1})$ is also a monotone Boolean function.

If  $g'$ is a constant function and  $g'(x_{-1})=0$, we have
$f(x)=\mbox{AND}_{k+1}(x_1x_2,\cdots,\linebreak[0] x_{k+1})$ and $Q_E(f)=k+1$.
Otherwise,  $\mbox{AND}_k(x_{-1})\leq g'(x_{-1})$, then  $f(x)=\mbox{AND}_k(x_{-1})$ and therefore   $Q_E(f)=k$.

{\bf Case 2:}  $Q_E(f_{x_1=0})=k$. There are again two subcases:

{\bf Case 2a:}  $f_{x_1=0}(x)=\mbox{OR}_k(x_{-1})$. Let us consider the DNF of $f$:
  \begin{equation}
    f(x)=\bigvee_{I\in D}\bigwedge_{i\in I} x_i=\left(\bigvee_{I\in D, 1\in I} \bigwedge_{i\in I} x_i \right)\vee\left(\bigvee_{I\in D, 1\not\in I}\bigwedge_{i\in I}x_i\right).
\end{equation}
We have
  \begin{equation}
    f(x)=(x_1\wedge h'(x_{-1}))\vee \mbox{OR}_{n-1}(x_{-1}),
\end{equation}
where $h'$ is a monotone Boolean function.  If  $h'$ is a constant function and  $h'(x_{-1})=1$,  then
$f(x)=\mbox{OR}_{k+1}(x_1x_2,\cdots,x_{k+1})$ and $Q_E(f)=k+1$. Otherwise  $h'(x_{-1})\leq OR_k(x_{-1})$ and therefore $f(x)=\mbox{OR}_k(x_{-1})$ and    $Q_E(f)=k$.

{\bf Case 2b:}  $f_{x_1=0}(x)=\mbox{AND}_k(x_{-1})$.  Let us consider the DNF of $f$.
  It has the form
    \begin{equation}
    f(x)=(x_1\wedge h_1(x_{-1}))\vee \mbox{AND}_k(x_{-1}),
\end{equation}
  where $h_1(x_{-1})$ is also a monotone Boolean function.
Therefore $f(x)=1$ for any $x$ such that $01\cdots 1\preceq x$ and $f(x)=0$ for any $x$ such that $x\prec 01\cdots 1 $.

Let us consider now two subcases: $f_{x_2=1}$ and $f_{x_2=0}$.
Since $0110\cdots 0\prec 01\cdots 1$, we have $f(0110\linebreak[0]\cdots 0)\linebreak[0]=0$ and $f_{x_2=1}(x)\neq \mbox{OR}_k(x_{-2})$.
Since $01\cdots 1\preceq 01\cdots 1$, we have $f(01\cdots 1)\linebreak[0]=1$ and $f_{x_2=1}(x)\neq \mbox{AND}_k(x_{-2})$.
Therefore we have $Q_E(f_{x_2=1})<k$ and   $Q_E(f_{x_2=0})\linebreak[0]=k$.
Since $0010\cdots 0\prec 01\cdots 1$, we have $f(0010\cdots 0)\linebreak[0]=0$ and $f_{x_2=0}(x)\neq \mbox{OR}_k(x_{-2})$.
Therefore,
 $f_{x_2=0}(x)=\mbox{AND}_k(x_{-2})$.
 Using a similar argument, we can prove that for any $i\geq 2$, $f_{x_i=0}(x)=\mbox{AND}_k(x_{-i})$.
 Hence, for any $i\in[k+1]$, we have
 \begin{equation}
    f(x)=(x_i\wedge h_i(x_{-i}))\vee \mbox{AND}_k(x_{-i}).
\end{equation}

Therefore $f(x)=1$ for any $x$ such that $y\preceq x$
 and $f(x)=0$ for any $x$ such that $x\prec y$, where $y_i=0$ and $y_j=1$ for any $j\neq i$. It is now not hard to show  that $f(x)=\mbox{Th}_{k+1}^{k}$ and $Q_E(f)=k$.

Therefore, the theorem has been proved.
\end{proof}

\section{Read-once Boolean functions}

\begin{theorem}\label{fsize}
If $f:\{0,1\}^n\to \{0,1\}$  is a read-once Boolean function, then $Q_E(f)=n$ iff $f$ is isomorphic to $\mbox{AND}_n$.
\end{theorem}

\begin{proof}
If $f$ is isomorphic to $\mbox{AND}_n$, then $Q_E(f)=n$ \cite{BBC+98}. We prove the other direction as follows.

Since $f$ is a read-once Boolean function,   $f$   depends on all $n$ variables and $L(f)=n$, i.e each $(\neg) x_i$  labels once and only once a leaf variable, where $(\neg)$ denotes a possible negation. We
prove the result by an induction.

{\bf BASIS}:  $\mbox{AND}_3(x_1,x_2,x_3)$ is the only $3$-bit Boolean function, up to isomorphism,  that requires 3 quantum queries \cite{MJM11}. Therefore the result holds for $n=3$.

{\bf INDUCTION}: We will suppose the result holds for all $n\leq k$ ($k\geq  3$) and we will prove that the result holds also for all $n\leq k+1$.

Suppose the root of a formula $F$ is labeled with $\wedge$. Without loss of generality, we assume that there exist Boolean functions $g:\{0,1\}^p\to \{0,1\}$ and $h:\{0,1\}^q\to \{0,1\}$ such that $f(x)=g(y)\wedge h(z)$ and $p+q=k+1$, where $x=yz$. Since $f$ depends on all $k+1$ variables and $L(f)=k+1$, we have $L(g)=p$ and $L(h)=q$, where $g$ depends on all $p$ variables and $h$ depends on all $q$ variables. If $Q_E(g)<p$ or $Q_E(h)<q$,  then $Q_E(f)\leq Q_E(g)+Q_E(h)<k+1$. Now suppose $Q_E(g)=p$ and  $Q_E(h)=q$. According to the assumption,
$g$ is isomorphic to  $\mbox{AND}_p$ and $h$ is isomorphic to $\mbox{AND}_q$. There are therefore the following four cases to consider.

{\bf Case 1:} $g(y)=\mbox{AND}_p\left((\neg)x_1,\ldots,(\neg)x_p\right)$ and $h(z)=\mbox{AND}_q\left((\neg)x_{p+1},\linebreak[0]\ldots,\linebreak[0](\neg)x_{k+1}\right)$.  Then $f$ is isomorphic to $\mbox{AND}_{k+1}$ and therefore $Q_E(f)=k+1$.

{\bf Case 2:} $g(y)=\mbox{OR}_p\left((\neg)x_1,\ldots,(\neg)x_p\right)$ and $h(z)=\mbox{OR}_q\left((\neg)x_{p+1},\ldots,\linebreak[0](\neg)x_{k+1}\right)$.  Therefore
\begin{equation}
  f(x)=\mbox{OR}_p\left((\neg)x_1,\ldots,(\neg)x_p\right)\wedge \mbox{OR}_q\left((\neg)x_{p+1},\ldots,(\neg)x_{k+1}\right).
\end{equation}
      Without loss of generality, we  suppose that $f(x)=\mbox{OR}_p\left(x_1,\ldots,x_p\right)\wedge \mbox{OR}_q\left(x_{p+1},\ldots,x_{k+1}\right)$. Since $p+k-p+1=k+1>3$, we have $p\geq 2$ or $k-p+1\geq 2$.    Without loss of generality, we assume that   $k-p+1\geq 2$. Let us query $x_2$ to $x_{k-1}$ first.
      \begin{enumerate}
        \item [1)] If  $x_i=1$ for some $2\leq i\leq p$ and $x_j=1$ for some $p+1\leq j\leq k-1$, then  $f_{x_2\cdots x_{k-1}}(x)=1$.
        \item [2)] If  $x_i=1$ for some $2\leq i\leq p$ and $x_{p+1}=\cdots=x_{k-1}=0$, then $f_{x_2\cdots x_{k-1}}(x)=\mbox{OR}_2\left(x_{k},x_{k+1}\right)$.
        \item [3)]  If $x_2=\cdots=x_p=0$ and $x_i=1$ for some $p+1\leq i\leq k-1$, then $f_{x_2\cdots x_{k-1}}(x)=x_1$.
        \item [4)] Otherwise, $x_2=\cdots=x_{k-1}=0$ and therefore $f_{x_2\cdots x_{k-1}}(x)=x_1\wedge(x_k\vee x_{k+1})$ and $Q_E(f_{x_2\cdots x_{k-1}})=2$.
      \end{enumerate}
            Therefore $Q_E(f)\leq k-2+2<k+1$.

{\bf Case 3:}   $g(y)=\mbox{AND}_p\left((\neg)x_1,\ldots,(\neg)x_p\right)$ and $h(z)=\mbox{OR}_q\left((\neg)x_{p+1},\ldots,(\neg)x_{k+1}\right)$. Therefore $f(x)=\mbox{AND}_p\left((\neg)x_1,\ldots,(\neg)x_p\right)\wedge \mbox{OR}_q\left((\neg)x_{p+1},\ldots,(\neg)x_{k+1}\right)$.  Without loss of generality, we can  now suppose that
\begin{equation}
 f(x)=\mbox{AND}_p\left(x_1,\ldots,x_p\right)\wedge \mbox{OR}_q\left(x_{p+1},\ldots,x_{k+1}\right).
  \end{equation}
  If $p=k$, then $f=\mbox{AND}_{k+1}$ and $Q_E(f)=k+1$. Now we  consider the case $p<k$. Let us query $x_2$ to $x_{k-1}$ first.
   \begin{enumerate}
        \item [1)] If $x_2\cdots x_{p}\neq 1\cdots1$, then $f(x)=0$.
        \item [2)]   If  $x_2\cdots x_{p}= 1\cdots1$ and $x_{p+1}\cdots x_{k-1}\neq 0\cdots0$, then $f_{x_2\cdots x_{k-1}}(x)=x_1$.
        \item [3)]   If $x_2\cdots x_{p}= 1\cdots 1$ and $x_{p+1}\cdots x_{k-1}= 0\cdots0$, then $f_{x_2\cdots x_{k-1}}(x)=x_1\wedge(x_k\vee x_{k+1})$ and $Q_E(f_{x_2\cdots x_{k-1}})=2$.
      \end{enumerate}
 Therefore  $Q_E(f)\leq k-2+2<k+1$.

{\bf Case 4:}  $g(y)=\mbox{OR}_p\left((\neg)x_1,\ldots,(\neg)x_p\right)$ and $h(z)=\mbox{AND}_q\left((\neg)x_{p+1},\ldots,(\neg)x_{k+1}\right)$. This case is analogous to the {\bf Case 3}.

Symmetrically,  we can  consider the case  that the root of the formula $F$ is labeled with $\vee$. In this case,  we will need to deal with  functions with the same structure of $f(x_1,x_2,x_3)=x_1\vee (x_2\wedge x_3)$, which is isomorphic to $x_1\wedge(x_2\vee x_3)$. We omit the details here.
\end{proof}

It is mentioned in \cite{San95} that  all $n$-bit read-once  Boolean functions have exact classical quantum complexity  $n$. We give now a rigorous proof of that:

\begin{theorem}\label{Th-readonce}
If $f:\{0,1\}^n\to \{0,1\}$ is a read-once Boolean function, then $D(f)=n$.
\end{theorem}
\begin{proof}
Let us consider the multilinear polynomial representation of $f$.
 It is easy to prove  by induction that $\textrm{deg}(f)=n$ and there is just one monomial of $f$ of the degree $n$.

 {\bf BASIS}:  If $n=1$, then $f(x)=(\neg) x_1$. Therefore, $\textrm{deg}(f)=1$.

{\bf INDUCTION}: Suppose the result holds for all $n\leq k$, we will prove the result holds for all $n\leq k+1$.

Without loss of generality, let us assume that three exists an $i\in [n]$ such that
\begin{equation}
f(x_1,\ldots,x_{k+1})=g(x_1,\ldots,x_i)\wedge h(x_{i+1},\ldots,x_{k+1})
\end{equation}
 or
\begin{equation}
 f(x_1,\ldots,x_{k+1})=g(x_1,\ldots,x_i)\vee h(x_{i+1},\ldots,x_{k+1}),
 \end{equation}
where $L(g)=i$, $L(h)=k+1-i$,  $g$ and $h$ depend on all their variables. According to assumption of the theorem, we have $\mbox{deg}(g)=i$  and $g(x_1,\ldots,x_i)=(\pm)\prod_{j=1}^i(\neg)x_j+p(x_1,\ldots,x_i)$ where $\mbox{deg}(p)<i$, and $\mbox{deg}(h)=k+1-i$  and $h(x_{i+1},\ldots,x_{k+1})=(\pm)\prod_{j={i+1}}^{k+1}(\neg)x_j+q(x_{i+1},\ldots,x_{k+1})$ where $\mbox{deg}(q)<k+1-i$.

Since
\begin{equation}
f(x_1,\ldots,x_{k+1})=g(x_1,\ldots,x_i)\wedge h(x_{i+1},\ldots,x_{k+1})=g\cdot h
 \end{equation}
and
\begin{equation}
f(x_1,\ldots,x_{k+1})=g(x_1,\ldots,x_i)\vee h(x_{i+1},\ldots,x_{k+1})=g+h-g\cdot h.
 \end{equation}
 Therefore $\textrm{deg}(f)=k+1$ and there is just one monomial of $f$ of the degree $k+1$.

  According to Lemma~\ref{D(f)-geq-Deg(f)}, $D(f)\geq \textrm{deg}(f)=n$. Thus, $D(f)=n$.
\end{proof}

\section{General $n$-bit Boolean functions}

In this section we prove  our main result.
Without explicitly pointed out, $n>3$ in this section.

If $f$ is an $n$-bit Boolean function that is isomorphic to $\mbox{AND}_{n}$, then there must exist $b=b_1\ldots b_n\in\{0,1\}^n$ such that  every $f_{x_i=b_i}$ is equivalent to $\mbox{AND}_{n-1}$ ($\mbox{OR}_{n-1}$) up to  some negations of variables. Moreover $b$ has to be unique.   For example, if $f(x)=\mbox{OR}_n(x_1,x_2,\ldots,x_n)$, then  we have $f_{x_i=0}(x)=\mbox{OR}_{n-1}(x_1,\ldots,x_{i-1},x_{i+1},\linebreak[0]\ldots,x_n)$  for $i\in[n]$ and   $b=0\ldots 0$.

For an $n$-bit Boolean function $f$ that has  exact quantum query complexity $n$, we prove the following lemma.

\vspace*{12pt}
\noindent
\begin{lemma}\label{Lm-c7-1}
Suppose that $\mbox{AND}_{n-1}$ is the only (n-1)-bit Boolean function,  up to isomorphism, has  exact quantum query complexity $n-1$.
Let $f:\{0,1\}^n\to \{0,1\}$ be an $n$-bit Boolean function  that has  exact quantum query complexity $n$. There exists one and only one $b=b_1\ldots b_n\in\{0,1\}^n$ for every $i\in[n]$ such that $f_{x_i=b_i}$ is   equivalent to  $\mbox{AND}_{n-1}$ ($\mbox{OR}_{n-1}$) up  to some negations of the variables.
\end{lemma}
\vspace*{12pt}
\noindent
{\bf Proof:}
In order to prove this lemma, we study some properties of exact quantum query complexity of Boolean functions.
According to  Eq.~(\ref{Eq-n-1ton}),  we have the following lemma:
\vspace*{12pt}
\noindent
\begin{lemma}\label{C6-l1}
Let $f:\{0,1\}^n\to \{0,1\}$ be a Boolean function. If there exists an $i\in [n]$ such that both $Q_E(f_{x_i=0})<n-1$ and $Q_E(f_{x_i=1})<n-1$, then $Q_E(f)<n$.
\end{lemma}

We know from \cite{MJM11} that $\mbox{AND}_3$    is the only $3$-bit Boolean function, up to isomorphism, that has exact quantum query complexity 3. For any $4$-bit function $f$,  if there exists   $i\in[4]$ such that   neither $f_{x_i=0}$ nor $f_{x_i=1}$ is  isomorphic to $\mbox{AND}_{n-1}$, then $Q_E(f)<4$.

\vspace*{12pt}
\noindent
\begin{lemma}\label{Lm-both}
Let $f:\{0,1\}^n\to \{0,1\}$ be a Boolean function. If there exists an $i\in [n]$ such that both $f_{x_i=0}$ and $f_{x_i=1}$ are isomorphic to $\mbox{AND}_{n-1}$, then $Q_E(f)<n$.
\end{lemma}
\vspace*{12pt}
\noindent
{\bf Proof:}
Without loss of generality, we can assume  that  $i=1$.    According to Eq.~(\ref{Eq-df(x)}), we have
 \begin{equation}
    f(x)=\left(\neg x_1\wedge f_{x_1=0}(x_2,\ldots,x_n)\right)\vee\left(x_1\wedge f_{x_1=1}(x_2,\ldots,x_n)\right).
\end{equation}
Suppose that at least one of the functions  $f_{x_1=0}$ and $f_{x_1=1}$ is equivalent to   $\mbox{AND}_{n-1}$ up to some negations of the variables. Without loss of generality, we will now assume that $f_{x_1=1}(x)=\mbox{AND}_{n-1}(x_2,\ldots,x_n)$.
To prove the theorem, we   consider two  cases.

  {\bf Case 1:} $f_{x_1=0}(x)=\mbox{AND}_{n-1}((\neg)x_2,\ldots,(\neg)x_n)$.
   In this case we have  two subcases.

 {\bf Case 1a:} $f_{x_1=0}(x)=\mbox{AND}_{n-1}(\neg x_2,\ldots,\neg x_n)$. We have $$f(x)=\mbox{AND}_{n}(\neg x_1,\neg x_2,\ldots,\neg x_n)\vee \mbox{AND}_{n}\linebreak[0](x_1,\linebreak[0]x_2,\ldots,x_n)=\neg \mbox{NAE}(\linebreak[0]x_1\linebreak[0]x_2,\linebreak[0]\ldots,x_n).$$ Therefore, $Q_E(f)<n$.

  {\bf Case 1b:} $f_{x_1=0}(x)\neq \mbox{AND}_{n-1}(\neg x_2,\ldots,\neg x_n)$. Without loss of generality, we can suppose that there exists a $k\in\{2,\ldots,n-1\}$ such that $f_{x_1=0}\linebreak[0](x)=\mbox{AND}_{n-1}(\neg x_2,\linebreak[0]\ldots,\neg x_k,x_{k+1},\ldots,x_n )$. Then  $$f(x)=\mbox{AND}_{n}(\neg x_1,\ldots,\neg x_k,x_{k+1}, \ldots,x_n)\vee \mbox{AND}_{n}\linebreak[0](x_1,\linebreak[0]x_2,\ldots,x_n)$$
       $$=\left(\mbox{AND}_{k}(\neg x_1,\ldots,\neg x_k)\vee \mbox{AND}_{k}\linebreak[0](x_1,\linebreak[0]\ldots,x_k)\right)\wedge \mbox{AND}_{n-k}(x_{k+1},\ldots,x_n)$$
       $$=\neg \mbox{NAE}_k(\neg x_1,\ldots,\neg x_k)\wedge \mbox{AND}_{n-k}(x_{k+1},\ldots,x_n).$$
       Therefore, $Q_E(f)< k+n-k=n$.

{\bf Case 2:} $f_{x_1=0}(x)=\mbox{OR}_{n-1}((\neg)x_2,\ldots,(\neg)x_n)$.
   This means that we have two subcases.

{\bf Case 2a:} $f_{x_1=0}(x)=\mbox{OR}_{n-1}(\neg x_2,\ldots,\neg x_n)$. If $g(y)= \mbox{AND}_{n-1}(x_{2},\ldots,x_n)$, then $$f(x)=\left(\neg x_1\wedge \neg g(y)\right)\vee \left( x_1\wedge g(y)\right)=x_1\oplus g(y).$$
      Therefore, $Q_E(f)<n$.

{\bf Case 2b:} $f_{x_1=0}(x)\neq \mbox{OR}_{n-1}(\neg x_2,\ldots,\neg x_n)$.  Without loss of generality, we can suppose that $f_{x_1=0}(x)=\mbox{OR}_{n-1}( x_2,(\neg)x_3 \linebreak[0]\ldots,(\neg) x_n )$, then let us query $x_2$ first. If $x_2=0$, then $f_{x_2=0}(x)=\neg x_1\wedge \mbox{OR}_{n-2}( (\neg)x_3 \linebreak[0]\ldots,(\neg) x_n )$. According to Theorem \ref{fsize}, $Q_E(f_{x_2=0})<n-1$. If $x_2=1$, then $f_{x_2=1}(x)=\neg x_1\vee \mbox{AND}_{n-1}\linebreak[0](x_1,\linebreak[0]x_3,\ldots,x_n)=\neg x_1\vee \mbox{AND}_{n-2}\linebreak[0](x_3,\ldots,x_n)$. According to Theorem \ref{fsize}, $Q_E(f_{x_2=1})<n-1$. According to Eq.~(\ref{Eq-n-1ton}), $Q_E(f)< n-1+1=n$.

Now we need to  consider the case  that both  $f_{x_1=0}$ and $f_{x_1=1}$ are  $\mbox{OR}_{n-1}$ functions. Without loss of generality, we assume that $f_{x_1=1}(x)=\mbox{OR}_{n-1}(x_2,\ldots,x_n)$.
This means that we have again two subcases.

{\bf Case 3a:} $f_{x_1=0}(x)=\mbox{OR}_{n-1}(x_2,\ldots,x_n)$. In this case, we have $f(x)=\mbox{OR}_{n-1}(x_2,\ldots,x_n)$ and $Q_E(f)=n-1<n.$

{\bf Case 3b:} $f_{x_1=0}(x)\neq \mbox{OR}_{n-1}(x_2,\ldots,x_n)$. Without loss of generality generality, let us suppose that there exists a $k\in\{2,\ldots,n\}$ such that $f_{x_1=0}(x)=\mbox{OR}_{n-1}(\neg x_2,\linebreak[0]\ldots,\neg x_k,\linebreak[0]x_{k+1},\ldots,x_n )$. In  such a case $$f(x)=\left(\neg x_1\wedge\mbox{OR}_{n-1}(\neg x_2,\ldots,\neg x_k,x_{k+1}, \ldots,x_n)\right)\vee \left(  x_1\wedge\mbox{OR}_{n-1}\linebreak[0](x_2,\ldots,x_n)\right)$$
    Let us query $x_{k+1}$ to $x_n$ first. If $x_{k+1}=\cdots=x_n=0$, let $g(y)=f(x_1,\ldots,x_k,0,\ldots,0)$, then $$g(y)=\left(\neg x_1\wedge\mbox{OR}_{n-1}(\neg x_2,\ldots,\neg x_k,)\right)\vee \left(  x_1\wedge\mbox{OR}_{n-1}\linebreak[0](x_2,\ldots,x_k)\right)$$
    $$= \mbox{NAE}_{n}(\neg x_1,x_2,\ldots,x_k).$$
    Therefore, $Q_E(g)<k.$
    Otherwise, there exists a $j\geq k+1$ such that $x_j=1$. It is now easy to show that $f(x)=\neg x_1\vee x_1=1$. Therefore, $Q_E(f)<n-k+k=n.$

\vspace*{12pt}
\noindent
\begin{lemma}\label{Lm-and-or}
Let $f:\{0,1\}^n\to \{0,1\}$ be a Boolean function.  If there exist an $i\in[n]$ such that $f_{x_i=b}$ is
 equivalent to $\mbox{AND}_{n-1}$ ($\mbox{OR}_{n-1}$)   up  to some negations of the variables, then   $f_{x_j=c}$ is not  equivalent to  $\mbox{OR}_{n-1}$ ($\mbox{AND}_{n-1}$)   up  to some negations of the variables for $j\neq i$,
where $b,c\in\{0,1\}$.
\end{lemma}
\vspace*{12pt}
\noindent
{\bf Proof:}
Without loss of generality, we assume that $i=1$, $j=2$ and  $f_{x_1=b}(x)=\mbox{AND}_{n-1}(x_2,\linebreak[0]\cdots,x_n)$. In such a case we have $f(bc00*\cdots*)=f(bc01*\cdots*)=0$\footnote{* will denote one bit that can be 0 or 1.}. If we fix $c$, then  there are more than one   inputs such that $f_{x_2=c}(x)=0$. Therefore, $f_{x_2=c}$ is  not  equivalent to  $\mbox{OR}_{n-1}$    up  to some negations of the variables.

{\bf Proof of Lemma \ref{Lm-c7-1}:} According to Lemma~\ref{C6-l1}, for every $i\in[n]$, there must exist a $b_i\in\{0,1\}$ such that $f_{x_i=b_i}$ is isomorphic to $\mbox{AND}_{n-1}$, otherwise $Q_E(f)<n$.  Without loss of generality, we assume that $f_{x_1=b_1}$ is equivalent to  $\mbox{AND}_{n-1}$ ($\mbox{OR}_{n-1}$)  up  to some negations of the variables. According to Lemma~\ref{Lm-and-or},  no $f_{x_i=b_i}$ is  equivalent to  $\mbox{OR}_{n-1}$ ($\mbox{AND}_{n-1}$) up  to some negations of the variables. Therefore, for every $i>1$,  $f_{x_i=b_i}$ is  equivalent to  $\mbox{AND}_{n-1}$ ($\mbox{OR}_{n-1}$)  up  to some negations of the variables.

Now, suppose there exists $c=c_1\ldots c_n\neq b$ for every $i\in[n]$ such that $f_{x_i=c_i}$ is equivalent to  $\mbox{AND}_{n-1}$ ($\mbox{OR}_{n-1}$) up  to some negations of the variables.
Since $c\neq b$, there exist $i\in[n]$ such that $b_i\neq c_i$. We have therefore that both  $f_{x_i=b_i}$ and $f_{x_i=c_i}$ are isomorphic to  $\mbox{AND}_{n-1}$.     According to Lemma~\ref{Lm-both}, we have $Q_E(f)<n$, which is a contradiction.

In order to make our main result easier to understand, we consider $4$-bit Boolean functions first.
\vspace*{12pt}
\noindent
 \begin{theorem}\label{Th6}
 If $f$ is a 4-bit Boolean function, then $Q_E(f)=4$ iff $f$ is isomorphic to $\mbox{AND}_4$.
\end{theorem}

\vspace*{12pt}
\noindent
{\bf Proof:}
If $f$ is isomorphic to $\mbox{AND}_4$, then $Q_E(f)=4$ \cite{BBC+98}.

Assume that a $4$-bit Boolean function $f$ such that $Q_E(f)=4$, we prove that $f$ is isomorphic to $\mbox{AND}_4$ as follows.
According to  Lemma \ref{Lm-c7-1}, there exists  one and only one $b=b_1b_2b_3b_4$ for every $i\in[4]$ such that $f_{x_i=b_i}$ is   equivalent to  $\mbox{AND}_{3}$ ($\mbox{OR}_{3}$) up  to some negations of the variables.
Since for any $4$-bit function  $f$ with $b=b_1b_2b_3b_4$, there exists a function $f'$ with $b'=0000$ isomorphic to $f$.  We can get $f'$  by some negations of the variables $x_i$ whenever $b_i=1$.
Therefore, without loss of generality, we assume that $b=0000$ and for every $i\in[4]$ such that $f_{x_i=0}$ is equivalent to  $\mbox{OR}_{3}$ up  to some negations of the variables.

There are three cases that we need now to consider:

\begin{table}
\caption{Values of  $4$-bit Boolean functions. }\label{T3}
\begin{center}
\begin{tabular}{|l|l|l|l|l|l|l|}
   \hline
  $x_1$ &$x_2$& $x_3$ &$x_4$ & $f(x)$: Case 1 & Case 2 & Case 3 \\
    \hline
  0& 0& 0& 0    & 0     & 1     & 1\\
    \hline
  0& 0& 0& 1    & 1     & 1     & 1\\
    \hline
  0& 0& 1& 0    & 1     & 1     & *\\
    \hline
  0& 0& 1& 1    & 1     & 1     &*\\
    \hline
  0& 1& 0& 0    &1      & 1     &1\\
    \hline
  0& 1& 0& 1    & 1     & 1     &1\\
    \hline
  0& 1& 1& 0    & 1     &1      &1\\
    \hline
  0& 1& 1& 1    & 1     & 0     &1\\
    \hline
      \hline
  1& 0& 0& 0    & 1     & 1     &1\\
    \hline
  1& 0& 0& 1    &  1    & 1     &1\\
    \hline
  1& 0& 1& 0    & 1     & 1     &1\\
    \hline
  1& 0& 1& 1    & 1     & 0     &1\\
    \hline
      \hline
  1& 1& 0& 0    &  1    & 1     &*\\
    \hline
  1& 1& 0& 1    &  1    &0      &* \\
    \hline
    \hline
  1& 1& 1& 0    & 1     &0      &*\\
    \hline
  1& 1& 1& 1    &*      &*       &*\\
  \hline
\end{tabular}

\end{center}
\end{table}

  {\bf Case 1:} For every $i\in[4]$, there is no negation variable occurrence in   $f_{x_i=0}$, that is $f_{x_1=0}(x)=\mbox{OR}(x_2,x_3,x_4)$, $f_{x_2=0}(x)=\mbox{OR}(x_1,x_3,x_4)$, $f_{x_3=0}(x)=\mbox{OR}(x_1,x_2,x_4)$ and $f_{x_4=0}(x)=\mbox{OR}(x_1,x_2,x_3)$. See Case 1 in Table \ref{T3} for values of $f(x)$. We still do not the value of $f(1111)$. If $f(1111)=1$, then $f(x)=\mbox{OR}(x_1,x_2,x_3,x_4)$,  which is isomorphic to  $\mbox{AND}_4$. If $f(1111)=0$, then $f(x)=\mbox{NAE}(x_1,x_2,x_3,x_4)$ and $Q_E(f)<4$.

  {\bf Case 2:} There are negations
    of all variables  in every  $f_{x_i=0}$, that is $f_{x_1=0}(x)=\mbox{OR}(\neg x_2, \neg x_3,\linebreak[0] \neg x_4)$, $f_{x_2=0}(x)=\mbox{OR}(\neg x_1, \neg x_3, \neg x_4)$, $f_{x_3=0}(x)=\mbox{OR}(\neg x_1, \neg x_2,\neg x_4)$ and $f_{x_4=0}(x)=\mbox{OR}(\neg x_1, \linebreak[0]\neg x_2, \neg x_3)$. See Case 2 in Table \ref{T3} for values of $f(x)$. If $f(1111)=1$, then $f(x)= \neg \mbox{Th}_4^3$ and $Q_E(f)=3<4$.  If $f(1111)=0$, then $f(x)= \neg \mbox{EXACT}_4^3$ and $Q_E(f)=3<4$.

    {\bf Case 3:} There is an $i\in[4]$ such that there is at least one negation variable occurrence and  one  no negation variable occurrence  in  $f_{x_i=0}$.
     Without loss of generality, we can now assume that $f_{x_1=0}(x)=\mbox{OR}( x_2, \neg x_3, (\neg) x_4)$. In order to analyse this case, we prove the following two lemmas first.
\vspace*{12pt}
\noindent
    \begin{lemma}\label{lm-7}
Let $f$ be an $n$-bit Boolean function and $f_{x_i=0}$ be    equivalent to  $\mbox{OR}_{n-1}$ up  to some negations of the variables for every $i\in[n]$. If $f_{x_1=0}(x)=\mbox{OR}_{n-1}(x_2,\neg x_3, (\neg)x_4,\ldots)$, then $f_{x_2=0}(x)=\mbox{OR}_{n-1}(x_1,\neg x_3, (\neg)x_4,\ldots)$ and $f_{x_3=0}(x)=\mbox{OR}_{n-1}(\neg x_1,\neg x_2, (\neg)x_4,\ldots)$.
\end{lemma}
\vspace*{12pt}
\noindent
{\bf Proof:}
Since $f_{x_1=0}(x)=\mbox{OR}_{n-1}(x_2,\neg x_3, (\neg)x_4,\ldots)$, there exists a $y\in\{0,1\}^{n-3}$ such that $f(001y)=0$.
Suppose that  $f_{x_2=0}(x)=\mbox{OR}_{n-1}(\neg x_1,(\neg) x_3, (\neg)x_4,\ldots)$ or $f_{x_2=0}(x)=\mbox{OR}_{n-1}((\neg) x_1, x_3, (\neg)x_4,\ldots)$. We have   $f(001y)=1$, which is a contradiction.
Therefore, $f_{x_2=0}=\mbox{OR}_{n-1}(x_1,\neg x_3, (\neg)x_4,\ldots)$.

Now suppose that $f_{x_3=0}(x)=\mbox{OR}_{n-1}(x_1, (\neg) x_2, (\neg)x_4,\ldots)$. There have to  exist $c\in\{0,1\}$ and $z\in\{0,1\}^{n-3}$ such that $f(0c0z)=0$.  Since $f_{x_1=0}(x)=\mbox{OR}_{n-1}(x_2,\neg x_3, (\neg)x_4,\ldots)$, we have $f(0c0z)=1$, which is a contradiction.
Suppose that $f_{x_3=0}(x)=\mbox{OR}_{n-1}((\neg) x_1, x_2, (\neg)x_4,\linebreak[0]\ldots)$. There exist $c\in\{0,1\}$ and $z\in\{0,1\}^{n-3}$ such that $f(c00z)=0$. Since $f_{x_2=0}(x)=\mbox{OR}_{n-1}(x_1,\neg x_3, \linebreak[0](\neg)x_4,\ldots)$, we have $f(c00z)=1$, which is a contradiction.
Therefore, $f_{x_3=0}(x)\linebreak[0]=\linebreak[0]\mbox{OR}_{n-1}(\neg x_1,\neg x_2,\linebreak[0] (\neg)x_4,\ldots)$.

\vspace*{12pt}
\noindent
\begin{lemma}\label{lm-8}
Let $f$ be an $n$-bit Boolean function. If there exist 4 distinct inputs $x,y,u,v\in\{0,1\}^n$ such that $f(x)=f(y)=1$ and $f(u)=f(v)=0$, then $f$ is not isomorphic to $\mbox{AND}_n$.
\end{lemma}
\vspace*{12pt}
\noindent
{\bf Proof:}
If $f$ is equivalent to $\mbox{AND}_n$ up to some negations of the variables, then there exists just one $x \in\{0,1\}^n$ such that $f(x)=1$.
If $f$ is equivalent to $\mbox{OR}_n$ up to some negations of the variables, then there exists just one $u \in\{0,1\}^n$ such that $f(u)=0$.

     According to Lemma \ref{lm-7}, we have $f_{x_2=0}(x)=\mbox{OR}(x_1, \neg x_3, (\neg) x_4)$, and $f_{x_3=0}(x)=\mbox{OR}(\neg x_1,\linebreak[0] \neg x_2, (\neg) x_4)$. See Case 3 in Table  \ref{T3} for values of $f(x)$. It is easy to see that if $x_1\oplus x_2=1$, then $f(x)=1$. If $x_1\oplus x_2=0$, then $x_1=x_2$ and $f$ can be represented as a $3$-bit Boolean function $g(x_2,x_3,x_4)$, see Table \ref{T4} for its values. Since  $f_{x_1=0}(x)=\mbox{OR}( x_2, \neg x_3, (\neg) x_4)$, we have either $g(010)=f(0010)=0$ or $g(011)=f(0011)=0$. Since $f_{x_3=0}(x)=\mbox{OR}(\neg x_1, \neg x_2, (\neg) x_4)$, we have either $g(100)=f(1100)=0$ or $g(101)=f(1101)=0$. We also have $g(000)=f(0000)$ and $g(001)=f(0001)=1$. According to Lemma \ref{lm-8}, $g(x_2,x_3,x_4)$ is not isomorphic to $\mbox{AND}_3$ and $Q_E(g)<3$.

     \begin{table}
\caption{Values of  $g(x_2,x_3,x_4)$.}\label{T4}
\begin{center}

\begin{tabular}{|l|l|l|l|l|l|l|}
  \hline
 $x_2$& $x_3$ &$x_4$ &  $g(x_2,x_3,x_4)$ \\
    \hline
   0& 0& 0       & 1\\
    \hline
   0& 0& 1       & 1\\
    \hline
   0& 1& 0       & *\\
    \hline
   0& 1& 1       &*\\
    \hline
   1& 0& 0       &*\\
    \hline
   1& 0& 1     &*\\
    \hline
   1& 1& 0       &*\\
    \hline
  1& 1& 1      &*\\
    \hline
\end{tabular}

\end{center}
\end{table}

     Now we give an exact quantum algorithm for $f$ as follows:
    \begin{enumerate}
      \item [1)] Evaluate  $x_1\oplus x_2$ with one query.
      \item [2)] If $x_1\oplus x_2=1$, then $f(x)=1$.
      \item [3)] If $x_1\oplus x_2=0$, then $f(x)=g(x_2,x_3,x_4)$. Evaluate $g$ with exact quantum  algorithm.
    \end{enumerate}
Therefore, we have $Q_E(f)<1+ Q_E(g)<1+3=4.$ The theorem has been proved.

Finally, we prove the most general case.  The main idea of the proof is similar to  the proof of the previous theorem.

\vspace*{12pt}
\noindent
 \begin{theorem}
 If $f$ is an $n$-bit Boolean function, then $Q_E(f)=n$ iff $f$ is isomorphic to $\mbox{AND}_n$.
\end{theorem}
\vspace*{12pt}
\noindent
{\bf Proof:}
If $f$ is isomorphic to $\mbox{AND}_n$, then $Q_E(f)=n$ \cite{BBC+98}.
We prove the other direction by an induction on $n$.

{\bf BASIS}: The result holds for $n=3$.

{\bf INDUCTION}: Suppose the result holds for $n-1$, we will prove that the result holds for $n$.
According to Lemma \ref{Lm-c7-1},
there exists  one and only one $b=b_1\ldots b_n$ for every $i\in[n]$ such that $f_{x_i=b_i}$ is   equivalent to  $\mbox{AND}_{n-1}$ ($\mbox{OR}_{n-1}$) up  to some negations of the variables.
Without loss of generality, we assume that $b=0\ldots0$ and for every $i\in[n]$ such that $f_{x_i=0}$ is equivalent to  $\mbox{OR}_{n-1}$ up  to some negations of the variables.

There are three cases that we need to consider:

  {\bf Case 1:} For every $i\in[n]$, there is no negation variable occurrence in  $f_{x_i=0}$, that is $f_{x_i=0}(x)=\mbox{OR}_{n-1}(x_1,\ldots,x_{i-1},x_{i+1},\ldots, x_n)$ for $i\in[n]$.
   It is easy to see that in such a case   $f(0\ldots0)=0$, $f(1\ldots1)=*$ and $f(x)=1$ for $x\not\in\{0\ldots0,1\ldots 1\}$. If  $f(1\ldots1)=1$, then $f(x)=\mbox{OR}_n(x_1,\ldots,x_n)$,   which is isomorphic to  $\mbox{AND}_n$.  If $f(1\ldots1)=0$, then $f(x)=\mbox{NAE}(x_1,\ldots,x_n)$ and $Q_E(f)<n$.

     {\bf Case 2:} There are all negation variable occurrences in every  $f_{x_i=0}$, that is $f_{x_i=0}(x)=\mbox{OR}_{n-1}(\neg x_1,\ldots,\linebreak[0]\neg x_{i-1},\neg x_{i+1},\ldots, \neg x_n)$ for $i\in[n]$. It is easy to see that $f(x)=1$ for $|x|<n-1$, $f(x)=0$ for $|x|=n-1$    and $f(x)=*$ for $|x|=n$.
      If $f(1\ldots1)=1$, then $f(x)= \neg \mbox{Th}_n^{n-1}$ and $Q_E(f)=n-1<n$.
       If $f(1\ldots1)=0$,  then $f(x)= \neg \mbox{EXACT}_n^{n-1}$ and $Q_E(f)=n-1<n$.

  {\bf Case 3:} There is an $i\in[n]$ such that there is at least one negation variable occurrence and  one  no negation variable occurrence  $f_{x_i=0}$.   Without loss of generality, we assume that $f_{x_1=0}(x)=\mbox{OR}( x_2, \neg x_3, \linebreak[0](\neg) x_4,\ldots)$. According to Lemma \ref{lm-7}, we have $f_{x_2=0}(x)=\mbox{OR}(x_1, \neg x_3,\linebreak[0] (\neg) x_4, \ldots)$ and $f_{x_3=0}(x)=\mbox{OR}(\neg x_1, \neg x_2, \linebreak[0] (\neg) x_4, \ldots)$.
  For any $y\in\{0,1\}^{n-2}$, $f(01y)=f(10y)=1$, that is  $f(x)=1$ if $x_1\oplus x_2=1$.  If $x_1\oplus x_2=0$, then $x_1=x_2$ and $f$ can be represented as an $(n-1)$-bit Boolean function $g(x_2,\ldots,x_n)$. Since   $f_{x_1=0}(x)=\mbox{OR}( x_2, \neg x_3, (\neg) x_4,\ldots)$, there must exist a $u\in\{0,1\}^{n-3}$ such that $f(001u)=g(01u)=0$.
   Since $f_{x_3=0}(x)=\mbox{OR}(\neg x_1, \neg x_2, (\neg) x_4, \ldots)$, there must exist a $v\in\{0,1\}^{n-3}$ such that $f(110v)=g(10v)=0$.
   We also have $g(00\ldots00)=f(000\ldots00)=1$ and $g(00\ldots01)=f(000\ldots01)=1$.
     According to Lemma \ref{lm-8}, we have  that $g(x_2,\ldots,x_n)$ is not isomorphic to $\mbox{AND}_{n-1}$ and $Q_E(g)<n-1$.

       Now we give an exact quantum algorithm for $f$ as follows:
    \begin{enumerate}
      \item [1)] Evaluate  $x_1\oplus x_2$ with one query.
      \item [2)] If $x_1\oplus x_2=1$, then $f(x)=1$.
      \item [3)] If $x_1\oplus x_2=0$, then $f(x)=g(x_2,\ldots,x_n)$. Evaluate $g$ with exact quantum algorithm.
    \end{enumerate}
Therefore, we have $Q_E(f)<1+ Q_E(g)<1+n-1=n.$ The theorem has been proved.

\vspace*{12pt}
\noindent
\begin{corollary}
Almost all $n$-bit Boolean functions can be computed by an exact quantum algorithm with less than $n$ queries.
\end{corollary}
\vspace*{12pt}
\noindent
{\bf Proof:}
It is easy to see that there are $2\times 2^n$ $n$-bit Boolean functions which are isomorphic to $\mbox{AND}_n$.   Since there are $2^{2^n}$ Boolean
functions on $n$ variables, we see that the fraction of  functions which have exact quantum query complexity $n$  is $o(1)$. Thus almost all $n$-bit Boolean functions can be computed by an exact quantum algorithm with less than $n$ queries.

\section{Conclusion}

We have first shown that  $\mbox{AND}_n$ is the only $n$-bit Boolean function in three special classes of Boolean functions, (including symmetric, monotone, read-once  functions),  up to isomorphism, that has exact quantum query complexity $n$. Finally, we have proved that  in general  $\mbox{AND}_n$ is the only  $n$-bit Boolean function, up to isomorphism, that has exact quantum query complexity $n$.
This shows that the advantages for exact quantum query algorithms are more common than previously thought.

 In the proof for special classes of Boolean functions, we have used their special properties of different types of Boolean functions.
  Each approach is different from each other. These approaches that we used in each type of  Boolean functions  may be helpful in analysis of exact quantum complexity for other interesting functions.
In the approach for general case, we have used the properties of the true value table of the Boolean functions.

\section*{Acknowledgements}
The authors are thankful to the anonymous referees   for their comments and suggestions on the early version of this paper.
The third author would like to thank
Alexander Rivosh for his help while visiting  University of Latvia.
Work of the first author
was supported by FP7 FET projects QCS and QALGO and ERC Advanced Grant MQC (at the University of Latvia) and by National Science Foundation under agreement No. DMS-1128155 (at IAS, Princeton). Any opinions, findings and conclusions or recommendations expressed in this material are those of the author(s) and do not necessarily reflect the views of the National Science Foundation.
Work of the second and third authors was supported by
the Employment of Newly Graduated Doctors of Science for Scientific Excellence project/grant (CZ.1.07./2.3.00/30.0009)  of Czech Republic.

\end{document}